\documentclass[submission,copyright,creativecommons]{eptcs}
\providecommand{\event}{DCM 2023}

\usepackage{graphicx,amsmath,amssymb}
\usepackage{amsthm}
\usepackage{pict2e,verbatim,nets}
\usepackage{alltt,multirow,threeparttable}
\usepackage{multicol}
\usepackage{color}

\theoremstyle{plain}
\newtheorem{theorem}{Theorem}[section]

\newtheorem{proposition}[theorem]{Proposition}
\newtheorem{lemma}[theorem]{Lemma}

\theoremstyle{definition}
\newtheorem{definition}[theorem]{Definition}

\newtheorem{example}[theorem]{Example}

\newcommand     \ar             {{\sf ar}}

\newcommand     \inter          {\bowtie}

\newcommand     \nested[1]        {\langle {#1} \rangle}

\newcommand{\ito}{\Rightarrow}         
\newcommand{\rto}{\rightarrow}         

\newcommand{\sym}[1]{{\mbox{\tt{}#1}}}

\newenvironment{program}{\begin{alltt}\small{}}{\end{alltt}}

\newcommand{\ie}{\textit{i}.\textit{e}.\ }
\newcommand{\normalscale}{0.72} 
\newcommand{\smallscale}{0.62} 
\newcommand{\tinyscale}{0.55} 

\title{Conditional Nested Pattern Matching in Interaction Nets}
\def\titlerunning{Conditional Nested Pattern Matching in Interaction Nets}

\author{Shinya Sato
\institute{Institute for Liberal Arts Education\\Ibaraki University, Ibaraki, Japan}}
\def\authorrunning{S. Sato}

\begin{document}
\maketitle
\bibliographystyle{eptcs}

\begin{abstract}
Interaction nets are a form of restricted graph rewrite system that can serve as a graphical or textual programming language. As such, benefits include one-step confluence, ease of parallelism and explicit garbage collection.  However, some of these restrictions burden the programmer, so they have been extended in several ways, notably to include data types and conditional rules. This paper introduces a further extension to allow nested pattern matching and to do so in a way that preserves these benefits and fundamental properties of interaction nets. We also show that by introducing a translation to non-nested matching, 
this extension is conservative in rewriting.
In addition, we propose a new notation to express this pattern matching.
\end{abstract}

\section{Introduction}\label{sec:introduction}

Interaction nets, proposed by Lafont~\cite{LafontY:intn}, can be
considered as an execution model of programming languages, where
programs are described as graphs and computation is realised by graph
reduction.  
They have been used for optimal
reduction~\cite{LampingJ:algolc,GonthierG:geoolr} and other efficient
implementations~\cite{MackieIC:yalyal} of $\lambda$-calculus,
the basis of functional programming languages.  
Indeed, as a programming language, interaction nets have several attractive features:
\begin{itemize}
    \item A simple graph rewriting semantics;
    \item A complete symmetry between constructors and destructors;
    \item They are Turing complete;
    \item Reduction steps are local, lending them 
    to parallel execution without amending the algorithm being executed;
    \item Memory management, including garbage collection, is explicitly part of the execution, improving speed and, again, not requiring any changes in a parallel environment. 
\end{itemize}

When writing programs in interaction nets, it is useful to have some
extensions to the basic net structure, to
facilitate the process.
It is analogous to PCF~\cite{Plotkin77},
which is an extension of the pure typed $\lambda$-calculus
obtained by adding some constants.
For example, ``pure'' interaction nets do not have 
built-in constants,
data types or conditional branching. 
Data types, 
such as integers, were introduced in~\cite{FMP:Combining}, and
conditional rewriting rules on values were introduced in
\cite{DBLP:phd/ethos/Sato15}. 

Although net reduction rules
have a basic, depth-one
pattern matching, a nested version has been introduced in
\cite{DBLP:journals/entcs/HassanS08} as a conservative extension, \ie
although it provides a new feature to the programmer,
it can be implemented using pure interaction nets and
thus retains 
fundamental properties of interaction nets such as 
the one-step confluence property (defined in Section~\ref{sec:interaction-nets-graph}).

The aim of this paper is to introduce conditional nested pattern matching on values, as a further extension of \cite{DBLP:journals/entcs/HassanS08}, 
whilst preserving one-step confluence.
We show that this extension is conservative by introducing 
a translation that maps the nested pattern matching back to pure nets.
We also propose a new notation for the pattern matching.
This notation is similar to the well-known case expressions in many programming languages.

One of the goals of this work is to show how to represent
functional programs (and more generally term rewriting systems) 
that contain nested pattern matching in interaction nets.
This is part of on-going work to demonstrate a real-world functional programming language that can take full advantage of interaction nets' built-in parallelism.

This paper is structured as follows. In the next section we give a overview of interaction nets and the extension for values. In Section~\ref{sec:cnaps} we introduce the conditional nested pattern matching.
Section~\ref{sec:translation}
introduces a translation that removes the nested pattern matching
and shows that this extended 
matching is a conservative extension.
Then Section~\ref{sec:notation} proposes
a notation to easily express the matching and 
Section~\ref{sec:discussion} discusses some implementation issues.
The paper concludes in Section~\ref{sec:conclusion} with an outlook for the future.

\section{Background}\label{sec:background}

In this section we review the interaction net paradigm, 
and describe some known extensions to it.

\subsection{Interaction nets}\label{sec:interaction-nets-graph}
Interaction nets are graph rewriting systems~\cite{LafontY:intn}.
Each node has a name $\alpha$ and one or more \emph{ports} 
which can be connected to ports of other nodes.
The number of ports is determined by the node name $\alpha$
and is called the node's \emph{arity}, written as $\ar(\alpha)$.
When $\ar(\alpha) = n$, then the node $\alpha$ has $n+1$ ports:
$n$ \emph{auxiliary} ports and a distinguished one called the
\emph{principal} port, labelled with an arrow.
Nodes are drawn as follows:
\begin{center}
\includegraphics[scale=\normalscale,keepaspectratio,clip]{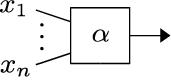}
\end{center}
We may put different labels on the ports to distinguish them.

We have a set $\Sigma$ of \emph{symbols}, which are the names of nodes.
A \emph{net} built on $\Sigma$ is an undirected graph:
the vertices are nodes, and the edges connect nodes at the ports. 
There is at most one edge at every port.
A port which is not connected is called a \emph{free} port.
A pair of nodes $(\alpha,\beta)\in \Sigma\times\Sigma$
connected via their principal ports forms an \emph{active} pair,
which is the interaction nets analogue of a redex pattern.
We refer to such a connected pair $(\alpha, \beta)$ as
$\alpha\inter\beta$.
A rule $((\alpha, \beta) \ito N)$ describes how to replace the pair $(\alpha, \beta)$
with the net $N$. 
$N$ can be any net as long as the set of free ports are preserved during the reduction.
The following diagram illustrates the idea:
\begin{center}
\includegraphics[scale=\normalscale,keepaspectratio,clip]{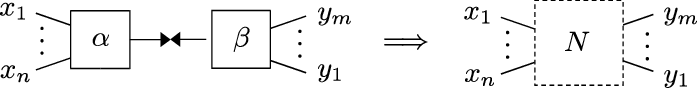}
\end{center}
There is at most one rule for each pair of nodes,
and the pairs are matched symmetrically, that is to say,
$(\beta, \alpha)$ is also replaced by $N$ by 
a rule $((\alpha, \beta) \ito N)$.
If the pair is symmetric, such as $(\alpha, \alpha)$, 
then $N$ must be symmetric.
We write $N_1 \rto N_2$
when a net $N_1$ is reduced to $N_2$ by a rule.
One of the most powerful properties of this graph
rewriting system is that it is one-step confluent;
if $N \rto N_1$ and $N \rto N_2$ $(N_1 \not= N_2)$,
then there exists a net $N'$ such that $N_1 \rto N'$ and $N _2 \rto N'$.
Therefore, if a normal form exists, it is uniquely determined and any reduction path to the normal form has the same number of steps.

Here,
as shown in Figure~\ref{fig:example-Mult-SZ},
we give an example of the interaction net system
for unary number multiplication, represented by
the following term rewriting system:
\begin{itemize}
    \item $\sym{Mult}(\sym{Z},y) = y$
    \item $\sym{Mult}(\sym{S}(x),y) = \sym{Add}(y,\sym{Mult}(x,y))$
\end{itemize}
This system has nodes $\sym{Z}$, $\sym{S}$, $\sym{Add}$ and $\sym{Mult}$
for arithmetic expressions and $\delta$, $\varepsilon$ for
duplication and elimination where
the $\alpha$ in the $\delta$ and $\epsilon$ rules
stands for either $\sym{Z}$ or $\sym{S}$.
This rewrites
$\sym{Mult}(\sym{S}(\sym{S}(\sym{Z})), \sym{S}(\sym{S}(\sym{S}(\sym{Z}))))$
to
$\sym{Add}(\sym{S}(\sym{S}(\sym{S}(\sym{Z}))), \sym{Add}(\sym{S}(\sym{S}(\sym{S}(\sym{Z}))), \sym{Z}))$.
Whilst in term rewriting systems
$\sym{Mult}(\sym{S}(x), A)$ can be rewritten as
$\sym{Add}(A, \sym{Mult}(x, A))$ even if the $A$ has redexes,
in this system the duplications are performed only for nets
that have no active pairs, \ie nets that are built from $\sym{S}$ and $\sym{Z}$.
We use $\delta$ to explicitly duplicate terms and $\epsilon$ to remove them when not needed in the result---an example of the explicit memory management.
In addition, evaluated results in the interaction net can be quickly used
for other computations without waiting for the whole computation to finish, \ie as a \emph{wait-free} algorithm.
This is called a \emph{stream operation}~\cite{Mackie-Sato-parallel},
and Figure~\ref{fig:example-Mult-SZ} shows
that in the second rewrite step,
the $\sym{S}$ duplicated by $\delta\inter\sym{S}$
can be used with $\sym{Add}$.

\begin{figure}[ht]
\begin{center}
  \includegraphics[scale=\smallscale,keepaspectratio,clip]
                {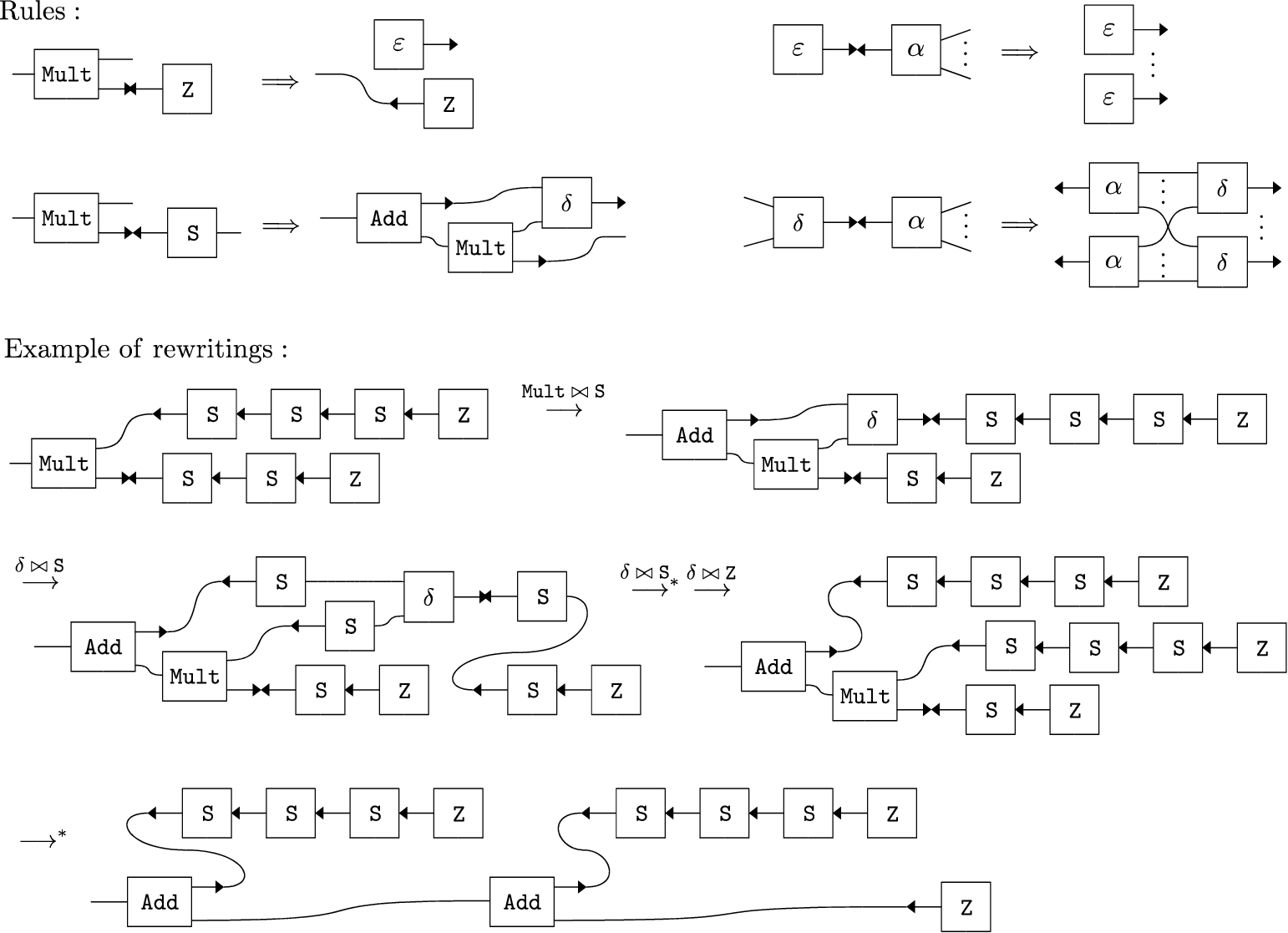}
\end{center}
\caption{An example of rules and rewritings of interaction nets}
\label{fig:example-Mult-SZ}
\end{figure}

We call the system \emph{pure} to mean that no extensions are applied.
From here, we show some extensions to the pure interaction nets system.

\subsection{Attributes}
Nodes can be extended to contain additional 
information~\cite{FMP:Combining,DBLP:phd/ethos/Sato15},
and here we review an extension called \emph{attributes}~\cite{DBLP:phd/ethos/Sato15}.
This extension is considered as an analogue of PCF obtained by adding some constants to the pure typed $\lambda$-calculus.
Attribute values are integers and are written in parentheses
following agent symbols.
For instance, $\alpha(2,4)$ is a node
which holds $2$ and $4$ as attribute values,
and is drawn as follows:
\begin{center}
\includegraphics[scale=\smallscale,keepaspectratio,clip]{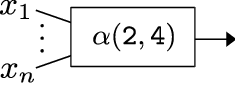}
\end{center}

\noindent
Integers and their lists are represented
using built-in nodes $\sym{Int}(i)$, $\sym{Cons}(j)$ and $\sym{Nil}$
where $i,j$ are integer numbers, and are drawn as follows:
\begin{center}
\includegraphics[scale=\smallscale,keepaspectratio,clip]{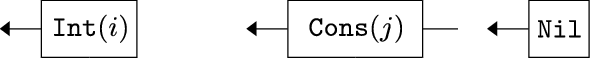}
\end{center}
To simplify the diagram,
the nodes $\sym{Int}(i)$ and $\sym{Cons}(j)$
are often drawn without the symbols when no confusion will happen.
For instance, an integer $\sym{1}$ 
and a list of $\sym{2}, \sym{4}, \sym{3}$ are written
in the following way:
\begin{center}
\includegraphics[scale=\smallscale,keepaspectratio,clip]{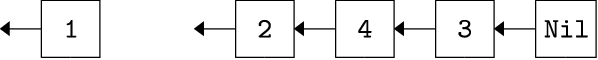}
\end{center}

Rewriting rules are also extended to deal with attributes
by considering a node name with attribute values as one symbol.
For instance, the name of the node $\alpha(\sym{2},\sym{4})$,
which holds attribute values $\sym{2}$ and $\sym{4}$,
is considered as a symbol ``$\alpha(\sym{2},\sym{4})$''.
We soon see that we need to add extra power to this system.
For example, the increment operator for integer numbers can be defined
by the following rules
between a node $\sym{Inc}$ and each integer number:
\begin{center}
\includegraphics[scale=\smallscale,keepaspectratio,clip]{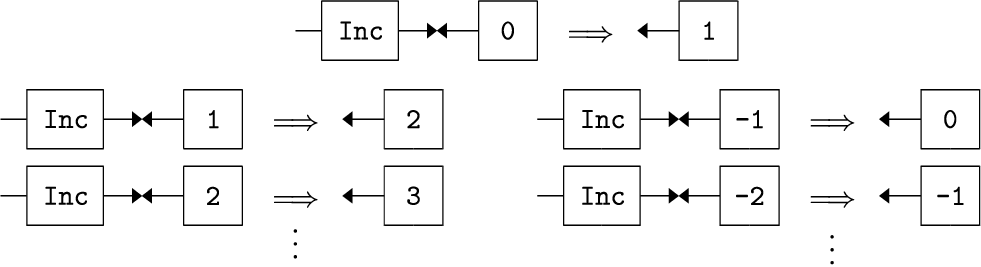}
\end{center}
This demonstrates an obvious problem in practical using---if we want to define the increment for any integer, we need not only an infinite set of symbols such as $\sym{Int}(\sym{0})$, $\sym{Int}(\sym{1})$, $\sym{Int}(\sym{-1})$, $\sym{Int}(\sym{2})$, $\ldots$, but also an infinite number of rules.
To deal with these in finite schemes,
we introduce meta variables $v$ for them, 
called \emph{attribute variables} and
\emph{expressions} $e$ on attribute variables as follows:

$e ::= v \, \mid \,
i \, \mid \,
\sym{-}e \, \mid \,
\sym{not} \ e  \, \mid \,
e \ \ \mathit{op} \ \ e
\, \mid \, (e)$

\noindent where $v$ is an attribute variable,
$i$ is an integer number, $e$ is an expression,
and $\mathit{op}$ is defined as follows:

$\mathit{op} ::= \sym{+} \, \mid \,
\sym{-} \, \mid \,
\sym{*} \, \mid \,
\sym{/} \, \mid \,
\sym{mod} \, \mid \,
\sym{==} \, \mid \,
\sym{!=} \, \mid \,
\sym{<} \, \mid \,
\sym{<=} \, \mid \,
\sym{>} \, \mid \,
\sym{=>} \, \mid \,
\sym{and} \, \mid \,
\sym{or}$.

\noindent{}The definition of expressions may be extended as long as the computation is decidable. 
We will abbreviate in the following the
left- and right-hand sides of a rule by LHS and RHS, respectively.
Attribute variables can be placed on the LHS nodes,
expressions on the variables can be placed on the RHS nodes, just like attributes.
The LHS nodes do not have the same attribute variables to ensure that each variable matches any attribute value; for example, $\alpha(i,i)$ is not allowed as the name of the LHS node because it has the same variable as $i$, whereas $\alpha(i,j)$ is allowed.
Now the increment operator is represented
as the following one rule, in which an attribute value $v$
is replaced by a value obtained by executing $v\texttt{+}1$:
\begin{center}
\includegraphics[scale=\smallscale,keepaspectratio,clip]{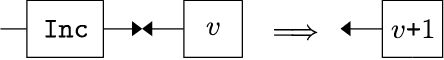}
\end{center}
Addition of integers on two node attributes is represented
as the following two rules:
\begin{center}
\includegraphics[scale=\smallscale,keepaspectratio,clip]{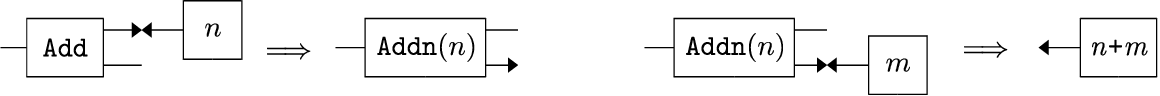}
\end{center}

\subsubsection{Conditional rules on attributes}
\emph{Conditional rules} were proposed in \cite{DBLP:phd/ethos/Sato15},
and here we introduce a refined version.
We write conditional rules as
$((\alpha, \beta) \ \mathrm{if}\ \mathit{c}\ {\ito} \ N)$,
where the $\mathit{c}$ is 
%
an expression on
attribute variables of $\alpha$, $\beta$.
The $\mathit{c}$ is called the \emph{conditional expression} of a rule,
and the rule becomes available to be applied
if $\mathit{c}$ is calculated to true,
where false is represented by 0 and true by any other value.
We also refer to the active pair with a conditional expression
as $(\alpha \inter \beta \ \mathrm{if}\ \mathit{c})$,
and call it a \emph{conditional} active pair.
The rule is drawn as follows:
\begin{center}
\includegraphics[scale=\normalscale,keepaspectratio,clip]{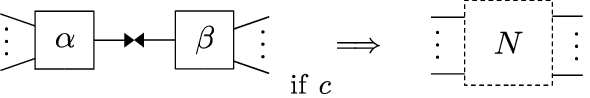}
\end{center}

\noindent
There may exist several conditional rules for the same active pair, 
but they must be \emph{disjoint},
\ie only at most one condition among the rules must be true.
This can be realised by introducing a decidable evaluation strategy for the conditions.
We use $\mathtt{true}$ and $\mathtt{false}$
that are evaluated as true and false, respectively.
In addition, as a refined version, we introduce a special conditional rule
$((\alpha, \beta) \ \mathrm{if} \ \mathit{otherwise} \ \ito N)$
which becomes available if any other conditional expressions
for the $\alpha\inter\beta$ become false.
We may write ``$\mathrm{if}\ \mathit{otherwise}$'' as just
``$\mathit{otherwise}$'' and omit ``$\mathrm{if}\ \mathtt{true}$''
when no confusion will arise.

For instance, a function $\sym{sumup}$
that takes a natural number
and computes the sum up to the number
can be realised by the following conditional rules:

\noindent
\begin{center}
\begin{minipage}[b]{5.0cm}
\begin{small}
\begin{program}
sumup 0 = 0
sumup n = n + sumup (n-1)
\end{program}
\end{small}
\end{minipage}
\qquad
\begin{minipage}{9.5cm}
\includegraphics[scale=\smallscale,keepaspectratio,clip]
                {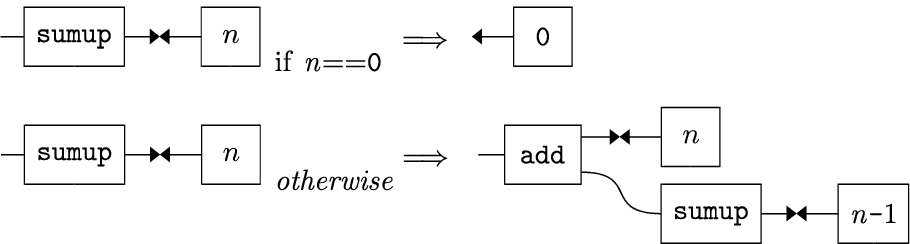}
\end{minipage}
\end{center}

Conditional rules are disjoint for the same active pairs,
therefore the one-step confluence property is preserved.

\section{Conditional Nested Pattern Matching}\label{sec:cnaps}
In this section we introduce conditional nested pattern matching
on attributes, as an extension of the nested matching
introduced in~\cite{DBLP:journals/entcs/HassanS08},
so that we can put a condition on each nested pattern.
We write \textsc{nap} as an abbreviation of ``nested active pair''.

Attribute values can be considered as parts of symbols,
so nested pattern matching can also manage attributes values.
Here, we extend the \textsc{nap} to conditional ones.
For simplicity, we consider all rules $((\alpha, \beta) \ito N)$
in the pure system as conditional rules with the $\mathtt{true}$ expression
such as $((\alpha, \beta) \ \mathrm{if}\ \mathtt{true}\ \ito N)$.

\begin{definition}[Conditional \textsc{nap}s]\label{def:conditional-naps}
  \emph{Conditional \textsc{nap}s} $\nested{P_\mathit{if}}$ are
  inductively defined as follows:
  
\begin{description}
\item[Base:] Every conditional active pair
  $(\alpha\inter\beta \ \mathrm{if}\ \mathit{c})$ is
  a conditional \textsc{nap}. 
  We write this as:
  \[\nested{\alpha(\vec{x})\inter\beta(\vec{y}) \ \mathrm{if}\ \mathit{c}}\]
  where $\vec{x}, \vec{y}$ are distinct names,
  corresponding to the occurrences of the auxiliary ports
  of $\alpha$ and $\beta$, respectively.

\item[Step:] We assume that $\nested{P_\mathit{if}}$ is
  a conditional \textsc{nap},
  $\gamma$ is an agent, $c$ is an expression
  on attribute variables of agents occuring in $\nested{P_\mathit{if}}$
  and a free port in $\nested{P_\mathit{if}}$ has a name $z$.
  A net obtained by connecting the principal port of $\gamma$
  with the expression $c$ to the free port $z$
  is also a conditional \textsc{nap}.
  \begin{center}
    \includegraphics[scale=\smallscale,keepaspectratio,clip]{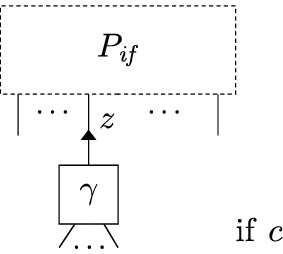}
  \end{center}
  We write this as:
  \[\nested{P_\mathit{if}, \ z-\gamma(\vec{w}) \ \mathrm{if}\ \mathit{c}}\]
  where $\vec{w}$ are distinct names and fresh for $\nested{P_\mathit{if}}$, corresponding to the occurrences of the auxiliary ports of $\gamma$.
  
  We call the form ``$z-\gamma(\vec{w}) \ \mathrm{if}\ \mathit{c}$''
  a \emph{connection} in a conditional \textsc{nap}. We use $u$ to range over connections.
\end{description}
\end{definition}

Expressions that occur
as ``$\mathrm{if}\ \mathit{c}$'' in conditional \textsc{nap}s are also called
\emph{conditional} expressions.
If all conditional expressions in a conditional \textsc{nap} evaluate to true,
the \textsc{nap} is called \emph{available}.

We build just a \textsc{nap} from a conditional \textsc{nap} by removing all conditional expressions.
For a conditional \textsc{nap} $\nested{P_\mathit{if}}$,
we write the (non-conditional) $\nested{P}$ 
as the \emph{condition dropped} \textsc{nap} for $\nested{P_\mathit{if}}$.
A rewriting rule on a conditional \textsc{nap}
$(\nested{P_\mathit{if}} \ito N)$ replaces
a matched net by the condition dropped $\nested{P}$ to the net $N$
if $\nested{P_\mathit{if}}$ is available.
Nested nets are also symmetrically matched and replaced.

\begin{example}\label{example:gcd}
We take the following program \texttt{gcd'}
which obtains the greatest common divisor of two given natural numbers:
\begin{program}
gcd' (a,b) = if b==0 then a
             else gcd' (b, a `mod` b)
\end{program}
This is written as the following
rewriting rules on conditional \textsc{nap}s:
\begin{center}
\includegraphics[scale=\smallscale,keepaspectratio,clip]{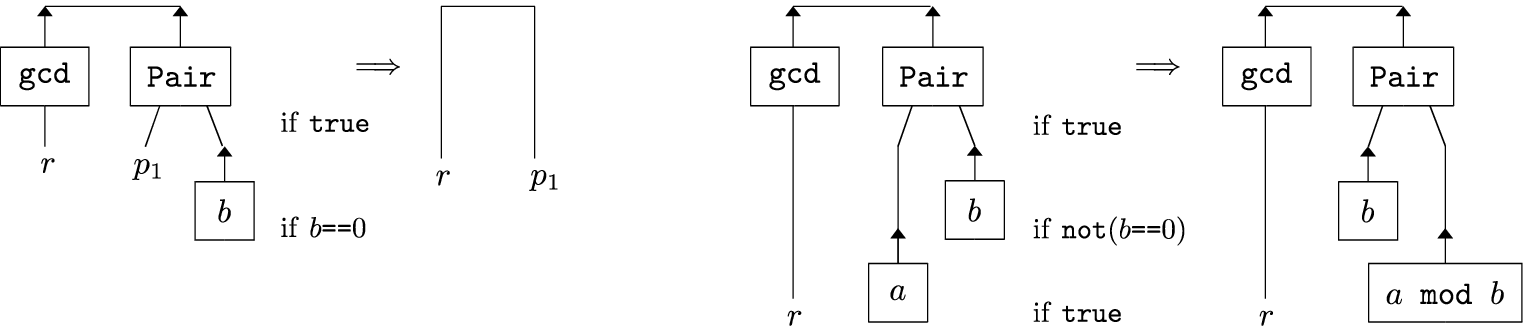}
\end{center}

\noindent{}
The following is an example of rewritings for \texttt{gcd'(21,14)}:
\begin{center}
\includegraphics[scale=\smallscale,keepaspectratio,clip]{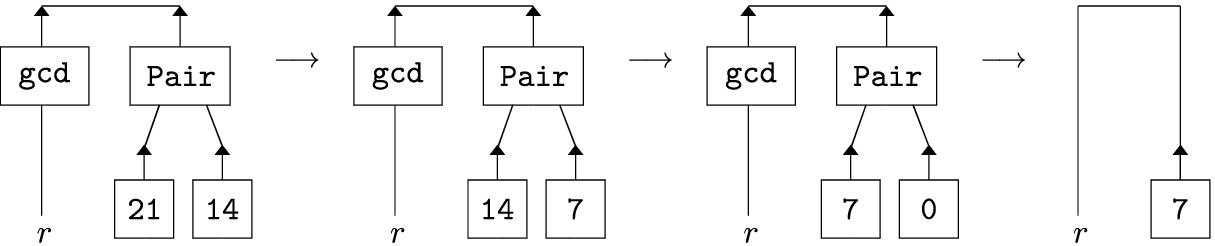}
\end{center}

\noindent{}
Of course, we can realise the same computation for the \texttt{gcd'} without
conditional \textsc{nap}s, but this requires several additional rules 
that obscure the overall meaning:

\begin{center}
\includegraphics[scale=\smallscale,keepaspectratio,clip]{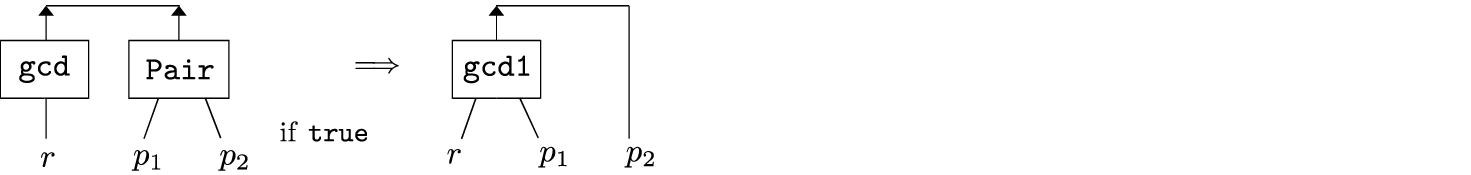}
\end{center}
\begin{center}
\includegraphics[scale=\smallscale,keepaspectratio,clip]{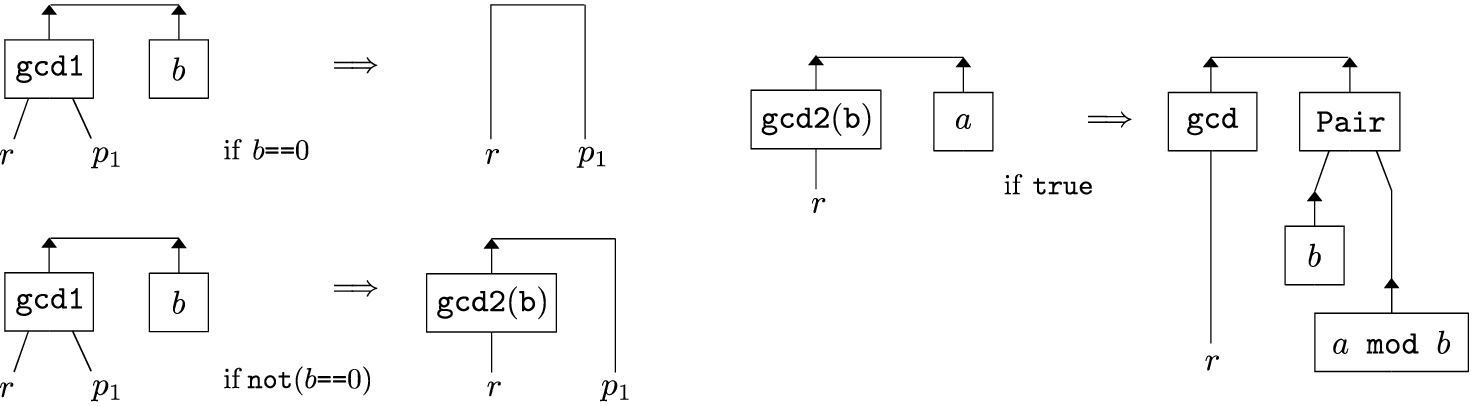}
\end{center}
\end{example}

\begin{example}\label{example:por}
Consider the following non-confluent map: \verb|f(0,y)=0|, \verb|f(x,0)=1|.

\noindent
This map is encoded as the following rules
where the agent \texttt{del} deletes built-in nodes $\sym{Int}(i)$:
\begin{center}
\includegraphics[scale=\smallscale,keepaspectratio,clip]{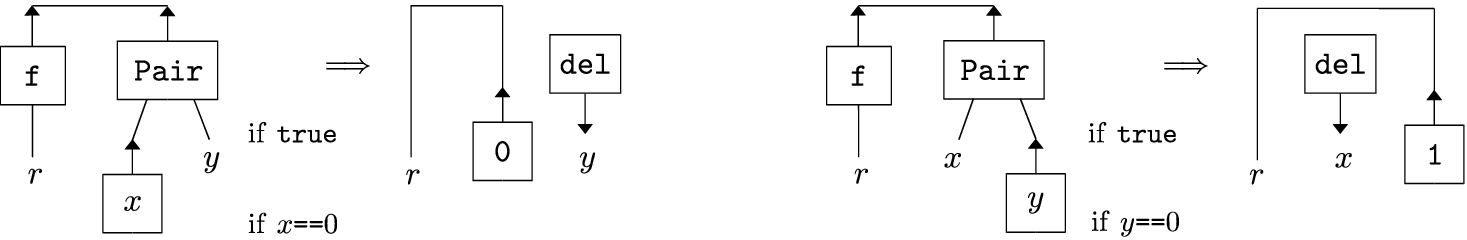}
\end{center}
As shown below, the two rewriting paths
from \texttt{f(0,0)} are not confluent:
\begin{center}
\includegraphics[scale=\smallscale,keepaspectratio,clip]{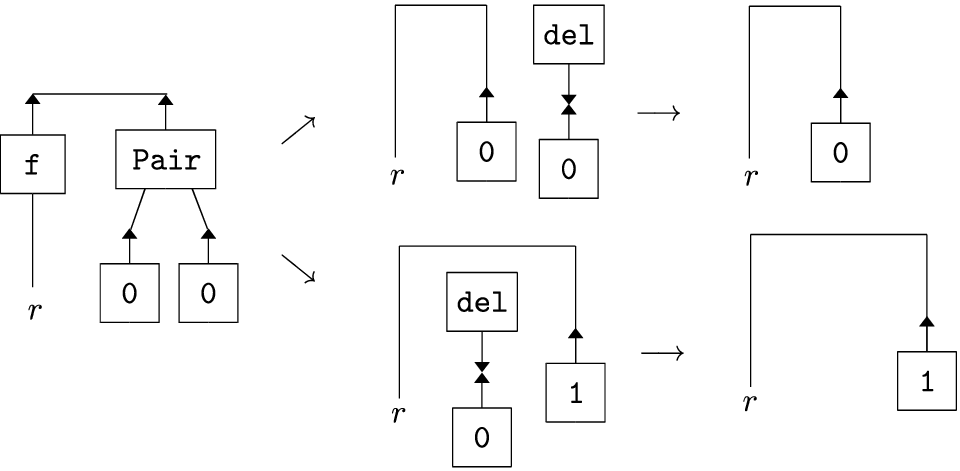}
\end{center}
  
\end{example}

To preserve the one-step confluence property, 
it is sufficient to introduce \emph{local sequentiality}~\cite{LafontY:intn},
\ie which port is looked first.
We extend several properties in \cite{DBLP:journals/entcs/HassanS08}.
First, we instantiate the sequentiality to sets of conditional \textsc{nap}s:

\begin{definition}[Local sequentiality]\label{definition:sequentiality-of-cnap}
  Let $\mathcal{P_\mathit{if}}$ be a set of conditional \textsc{nap}s.
  $\mathcal{P_\mathit{if}}$ is \emph{local sequential} if:

\begin{description}
\item[(1)] If $\nested{\alpha(\vec{x})\inter\beta(\vec{y}) \ \mathrm{if}\ \mathit{c}} \in \mathcal{P_\mathit{if}}$:
  \begin{description}
  \item[(1a)] All conditional expressions $c_k$ such that 
  $\nested{\alpha(\vec{x})\inter\beta(\vec{y}) \ \mathrm{if}\ \mathit{c_k}} \in \mathcal{P_\mathit{if}}$
    are disjoint, \ie only at most one conditional expression must be true.
  \end{description}
\item[(2)] If $\nested{P_\mathit{if}, \ z-\gamma(\vec{w}) \ \mathrm{if}\ \mathit{c}} \in \mathcal{P_\mathit{if}}$:
  \begin{description}
  \item[(2a)] All conditional expressions $c_k$ such that 
  $\nested{P_\mathit{if}, \ z-\gamma(\vec{w}) \ \mathrm{if}\ \mathit{c_k}} \in \mathcal{P_\mathit{if}}$
    are disjoint,
  \item[(2b)] $\nested{P_\mathit{if}} \in \mathcal{P_\mathit{if}}$,
  \item[(2c)]
    $\nested{P_\mathit{if}, \ z_j-\gamma_j(\vec{w_j}) \ \mathrm{if}\ \mathit{c_j}} \not\in \mathcal{P_\mathit{if}}$
    for any free port $z_j$ in the $\nested{P_\mathit{if}}$ except the $z$,
    where $\gamma_j$ is an agent,  $c_j$ is a conditional expression. 
  \end{description}
\end{description}
\end{definition}

We may write ``local sequential'' as just ``sequential''
when no confusion will arise.
If a local sequential set has an element
$\nested{\alpha(\vec{x})\inter\beta(\vec{y}) \ \mathrm{if}\ \mathit{c},\
  \vec{u}, \ z-\gamma(\vec{w}) \ \mathrm{if}\ \mathit{c}}$ 
  where $\vec{u}$ is a sequence of connections,
then every element started from the origin pair $\alpha\inter\beta$
is along the path to the $z$.
For clarity, we draw triangles on auxiliary ports connected to agents,
and horizontal dotted lines for the sequential order.
As an example, we depict a \textsc{nap} which is textually represented
$\nested{\alpha(\vec{x})\inter\beta(\vec{y}) \ \mathrm{if}\ \mathit{c}_1,
  \ z_k-\kappa(\vec{w}) \ \mathrm{if}\ \mathit{c}_2, 
  \ z_1-\gamma(\vec{z}) \ \mathrm{if}\ \mathit{c}_3}$ as follows:
\begin{center}
\vspace{-2mm}
\includegraphics[scale=\smallscale,keepaspectratio,clip]{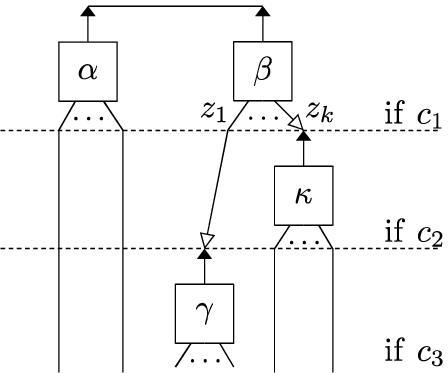}
\end{center}

\begin{example}\label{example:lastElt}
The following program \texttt{lastElt}
returns the last element of a given non-empty list:
\begin{program}
lastElt ([x]) = x
lastElt (x:y:ys) = lastElt (y:ys)
\end{program}
This is written as the following
rewriting rules on (non-conditional) \textsc{nap}s
where $\texttt{del}$ erases any agent:
\begin{center}
\includegraphics[scale=\tinyscale,keepaspectratio,clip]{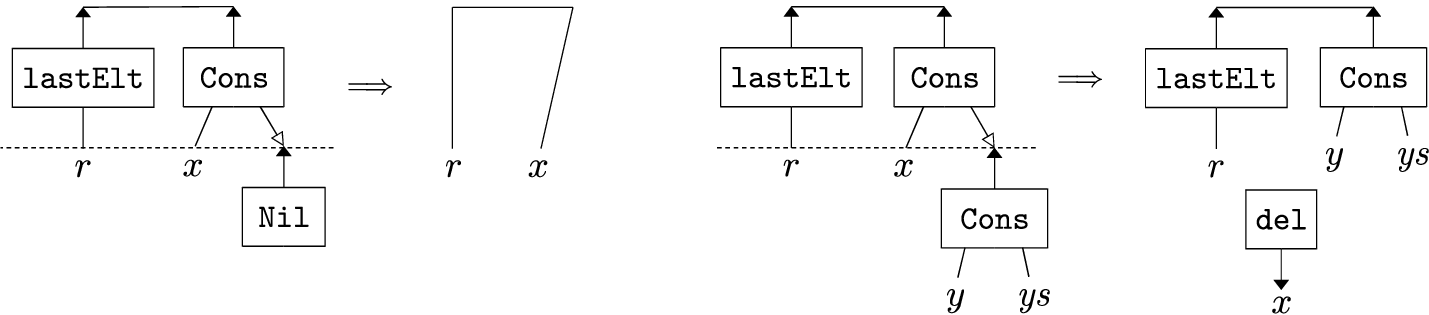}
\end{center}
These are written as follows (omitting ``$\mathrm{if} \ \mathtt{true}$'') and the set of the LHS nodes of the rules is sequential:
\begin{itemize}
\item $(\nested{\texttt{lastElt}(r)\inter\texttt{Cons}(x, \mathit{xs}),
  \ \ \mathit{xs} - \mathtt{Nil}}
  \ito N_1)$, 
\item $(\nested{\texttt{lastElt}(r)\inter\texttt{Cons}(x,\mathit{xs}),
  \ \ \mathit{xs} - \mathtt{Cons}(y, \mathit{ys})} \ito N_2)$.
\end{itemize}
\end{example}

\begin{example}
The rewriting rules in Example~\ref{example:gcd}
are textually written and graphically drawn as follows
and the set of the LHS nodes of the rules is also sequential:
\begin{itemize}
\item $(\nested{\texttt{gcd}(r)\inter\texttt{Pair}(p_1, p_2) \ \mathrm{if} \ \mathtt{true},
  \ \ p_2 - \mathtt{Int}(b) \ \mathrm{if} \ b \sym{==} 0}
  \ito N_1)$,
\item $(\nested{\texttt{gcd}(r)\inter\texttt{Pair}(p_1,p_2) \ \mathrm{if} \ \mathtt{true},
  \ \ p_2 - \mathtt{Int}(b) \ \mathrm{if} \ \sym{not}(b \sym{==} 0),
  \ \ p_1 - \mathtt{Int}(a) \ \mathrm{if} \ \mathtt{true}} \ito N_2)$.
\end{itemize}
\begin{center}
\includegraphics[scale=\tinyscale,keepaspectratio,clip]{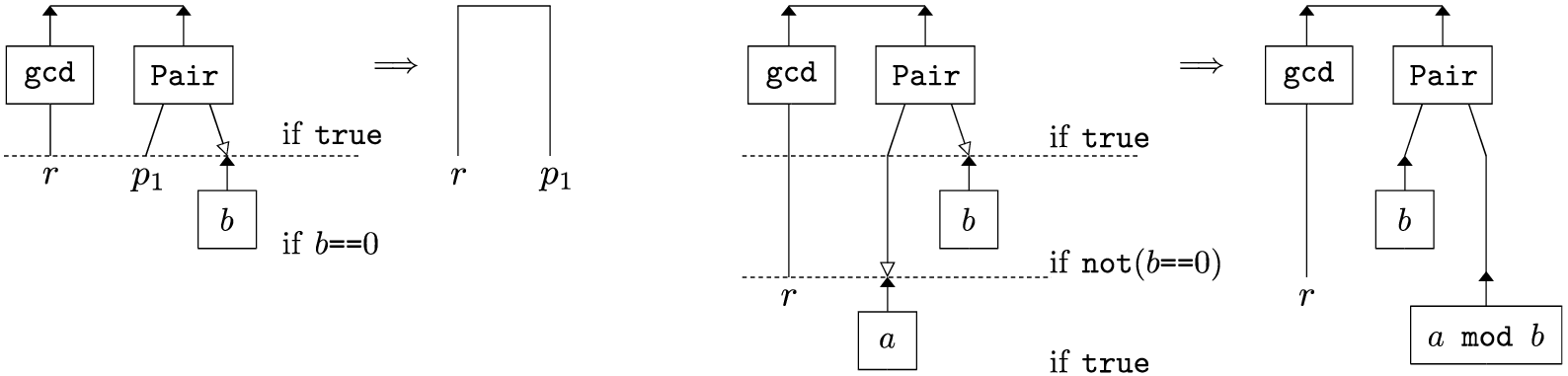}
\end{center}
\end{example}

Next, we extend well-formed rule sets~\cite{DBLP:journals/entcs/HassanS08}
as sets of pairwise distinct rules on conditional \textsc{nap}s:

\begin{definition}[Subnets of conditional \textsc{nap}s]
  Let $\nested{P_\mathit{if}}$ and $\nested{P'_\mathit{if}}$ be
  conditional \textsc{nap}s.
  $\nested{P_\mathit{if}}$ is a \emph{subnet} of $\nested{P'_\mathit{if}}$ if:
  \begin{itemize}
  \item $\nested{P}$ is a subnet of $\nested{P'}$,
    where $\nested{P}$ and $\nested{P'}$ are condition dropped \textsc{nap}s for
    $\nested{P_\mathit{if}}$ and $\nested{P'_\mathit{if}}$, respectively.
  \item $\nested{P_\mathit{if}}$ is available if $\nested{P'_\mathit{if}}$ is available.
  \end{itemize}
\end{definition}

\begin{definition}[Sets of pairwise distinct rules]\label{def:well-formed-condition-set}
  A set of \emph{pairwise distinct} rules on conditional \textsc{nap}s
  $\mathcal{R_\mathit{if}}$ is defined if:
\begin{itemize}
\item There is a local sequential set which contains
  every conditional \textsc{nap}
  $\nested{P_\mathit{if}}$ such that
  $(\nested{P_\mathit{if}} \ito N) \in\mathcal{R_\mathit{if}}$,
\item For any $(\nested{P_\mathit{if}} \ito N) \in \mathcal{R_\mathit{if}}$, \
  $\mathcal{R_\mathit{if}}$ has no
  $(\nested{P'_\mathit{if}} \ito N')$
  where
  $\nested{P'_\mathit{if}}$ is a subnet of $\nested{P_\mathit{if}}$.
\end{itemize}
\end{definition}
\noindent{}We call a rule set \emph{pairwise distinct} if the set satisfies conditions in Definition~\ref{def:well-formed-condition-set}.
We also call rules \emph{pairwise distinct}
if there is a set of pairwise distinct rules containing them.

As long as a rule set is pairwise distinct,
one-step confluence property is preserved as follows:

\begin{proposition}[One-step confluent]
  When a given set of rules on conditional \textsc{nap}s
  $\mathcal{R_\mathit{if}}$ is pairwise distinct,
  then all reductions using rules in $\mathcal{R_\mathit{if}}$
  are confluent in one step.
\end{proposition}

\begin{proof}
  We assume that $\mathcal{R_\mathit{if}}$ has two rules
  $(\nested{P_\mathit{if}} \ito N)$ and
  $(\nested{P'_\mathit{if}} \ito N')$
  which rewrite the same \textsc{nap} into different nets
  $N$ and $N'$.
  Let $\nested{P}$ and $\nested{P'}$ be condition dropped \textsc{nap}s
  for $\nested{P_\mathit{if}}$ and $\nested{P'_\mathit{if}}$, respectively.
  
  First, we suppose that
  $\nested{P}$ and $\nested{P'}$ are the same.
  If both $\nested{P_\mathit{if}}$ and $\nested{P'_\mathit{if}}$ are
  available,
  then conditional expressions
  in $\nested{P_\mathit{if}}$ and $\nested{P'_\mathit{if}}$
  are not disjoint, and this contradicts that 
  $\nested{P_\mathit{if}}$ and $\nested{P'_\mathit{if}}$
  are elements of the same sequential set.
  Otherwise,
  at most only one rule can be applied to the $\nested{P}$,
  and thus no critical pairs will be created by using both rules.
  
  Next, we suppose that
  the $\nested{P}$ is a subnet of $\nested{P'}$.
  If both $\nested{P_\mathit{if}}$ and $\nested{P'_\mathit{if}}$ are
  available,
  then $\nested{P_\mathit{if}}$ is a subnet of $\nested{P'_\mathit{if}}$,
  and it contradicts that
  the $\mathcal{R_\mathit{if}}$ is well-formed.
  Otherwise, there is at most only one rule
  that can be applied to both $\nested{P}$ and $\nested{P'}$.
  As a result, there are no critical pairs by using both rules.
\end{proof}

\section{Translation of nested conditional rules into non-nested ones}\label{sec:translation}
In this section we introduce a translation
$\mathbf{T}$ of a rewriting rule on a conditional \textsc{nap} 
$(\nested{P_\mathit{if}} \ito N)$
into non-nested conditional rules and show properties of the translation $\mathbf{T}$.
Generally nested pattern matching can be unfolded by introducing fresh symbols into less nested ones~\cite{10.1007/BFb0026099}.
In \textsc{nap}s each nested agent is inductively connected, 
and thus 
by considering such induction steps as the local sequential order,
we realise the nested matching 
as non-nested ones by introducing fresh agents.

\begin{definition}
We define a translation $\mathbf{T}$ of a rewriting rule on a conditional \textsc{nap} 
$(\nested{P_\mathit{if}} \ito N)$
into non-nested conditional rules
by structural induction on conditional \textsc{nap}s:

\begin{description}
\item[Base case:] If the $\nested{P_\mathit{if}}$ is
  $\nested{\alpha(\vec{x})\inter\beta(\vec{y}) \ \mathrm{if}\ \mathit{c}}$,
  then the translation $\mathbf{T}$ works as the identity function.

\item[Step case:] If the $\nested{P_\mathit{if}}$ is 
  $\nested{\alpha(\vec{x})
  \inter
  \beta(\vec{y}) \ \mathrm{if}\ \mathit{c}, \ z-\gamma(\vec{w}) \ \mathrm{if}\ \mathit{c}_z,\ \vec{u}}$
  where $z$ is an auxiliary port in $\alpha(\vec{x})\inter\beta(\vec{y})$
  and $\vec{u}$ is a sequence of connections
  such as 
  ``$z_1-\gamma(\vec{w_1}) \ \mathrm{if}\ \mathit{c}_1,\ 
  z_2-\gamma(\vec{w_2}) \ \mathrm{if}\ \mathit{c}_2, \ \ldots$''.

  To draw graphs simply, we suppose that $\vec{y}$ is $y_1, \ldots, y_m, y_{m+1}$
  and the $z$ is $y_{m+1}$.
  Other cases can be defined in a similar way.
  We also write the $\alpha$ and $\beta$ as
  $\alpha'(\vec{a})$ and $\beta'(\vec{b})$, respectively,
  in order to show the attribute variables explicitly.

  The $\mathbf{T}$ generates the following rules:
  \begin{itemize}
  \item $(\alpha'(\vec{a})(\vec{x})\inter\beta'(\vec{b})(\vec{y}) \ \mathrm{if}\ \mathit{c} \ito M)$
    where $M$ is a net obtained by connecting the $y_{m+1}$ to
    the principal port of an agent 
    $\kappa(\vec{a},\vec{b})(\vec{x},y_1, \ldots, y_m)$
    and the $\kappa$ is fresh for $\Sigma$ and the previous occurring symbols
    and depending uniquely on
    the symbols $\alpha', \beta'$ and the conditional expression $c$,

  \item $\mathbf{T}[(\nested{\kappa(\vec{a},\vec{b})(\vec{x}, y_1,\ldots,y_m)
    \inter \gamma(\vec{w}) \ \mathrm{if}\ \mathit{c}_z,\ \vec{u}}
    \ito N)]$.
    \smallskip

  \bigskip
  \includegraphics[scale=\smallscale,keepaspectratio,clip]{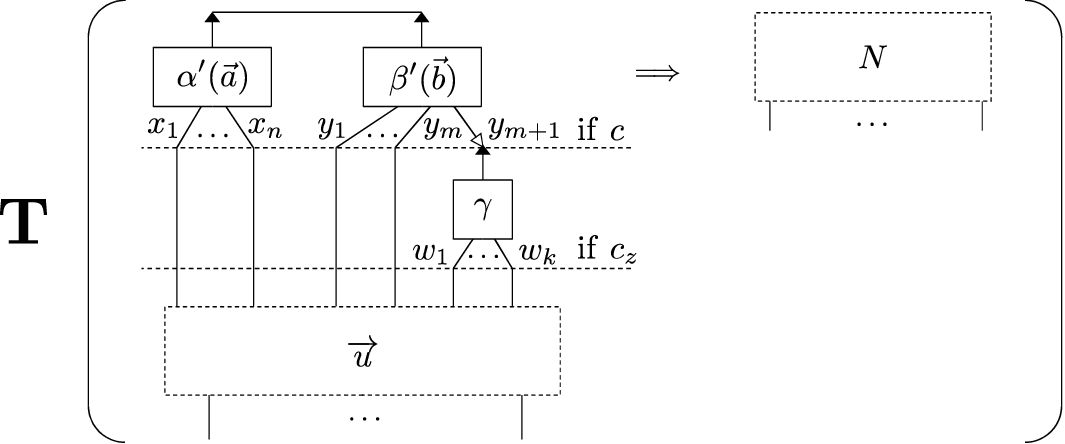}

  \bigskip

  \includegraphics[scale=\smallscale,keepaspectratio,clip]{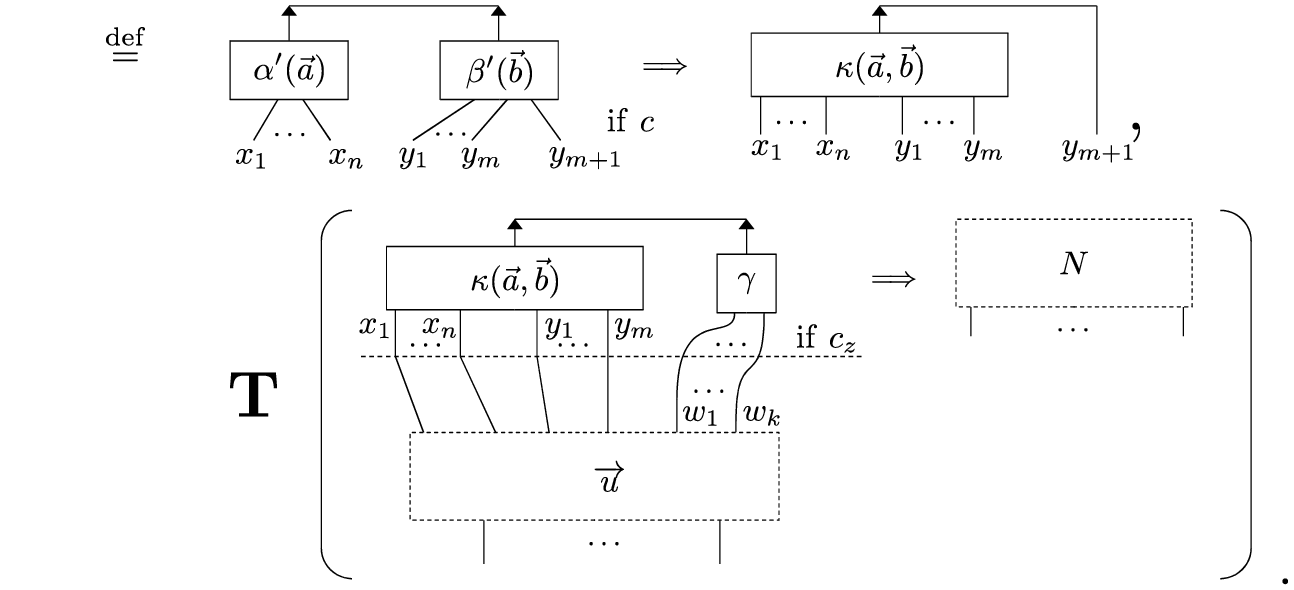}
  \end{itemize}

\end{description}
\end{definition}

\begin{example}\label{example:trans-gcd}
We show the translation of the rules for
$\texttt{gcd}\inter\texttt{Pair}$ in Example~\ref{example:gcd}.
First, we take the following rule:
$(\nested{\texttt{gcd}(r)\inter\texttt{Pair}(p_1,p_2) \ \mathrm{if} \ \mathtt{true},
  \ \ p_2 - \mathtt{Int}(b) \ \mathrm{if} \ b\mathtt{==}0}
  \ito N_1)$.
%
By applying the translation  to the rule,
we have the following rule and translation:
\begin{center}
\includegraphics[scale=\tinyscale,keepaspectratio,clip]{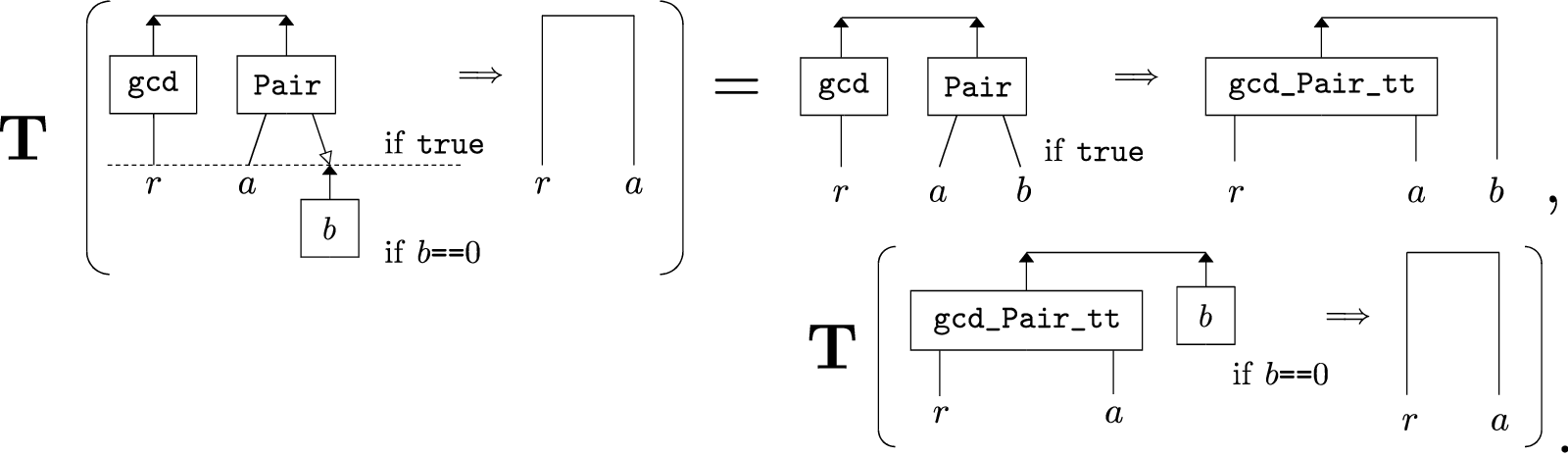}
\end{center}

\noindent
$\mathrm{T}$ works for 
$\nested{\mathtt{gcd\_Pair\_tt}\inter \mathtt{Int}(b) \ \mathrm{if} \ \mathtt{true}}$ as the identity function, so we obtain the following rules as the result: 
\begin{center}
\includegraphics[scale=\tinyscale,keepaspectratio,clip]{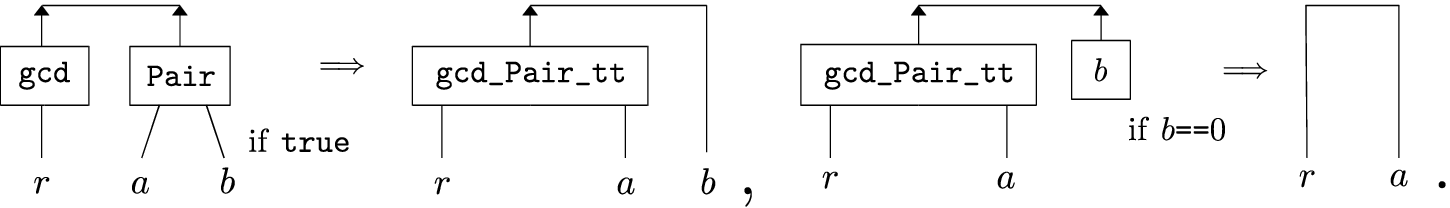}
\end{center}

\noindent
The other rule for $\texttt{gcd}\inter\texttt{Pair}$
is similarly translated as follows:
\medskip

\noindent
\includegraphics[scale=\tinyscale,keepaspectratio,clip]{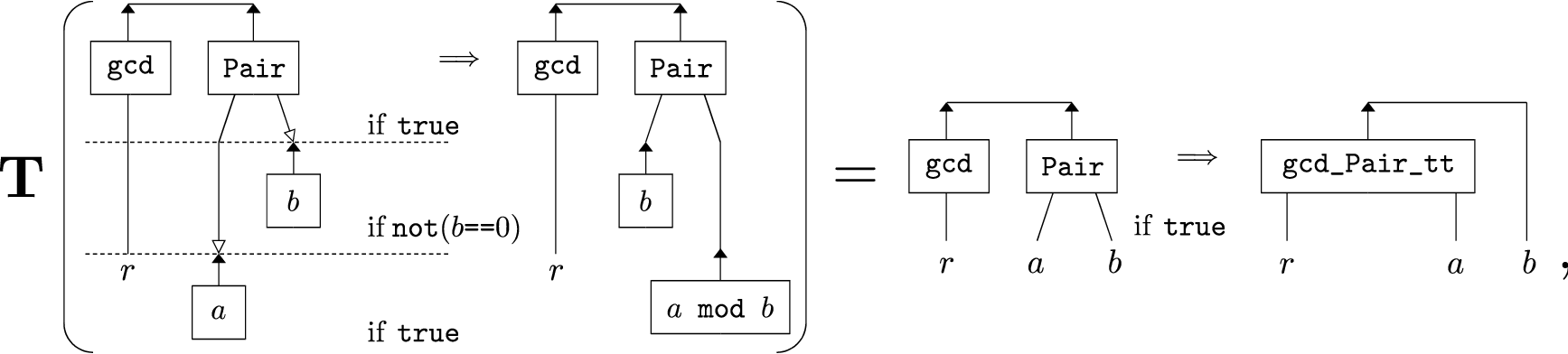}
\medskip

\noindent
\ \ \includegraphics[scale=\tinyscale,keepaspectratio,clip]{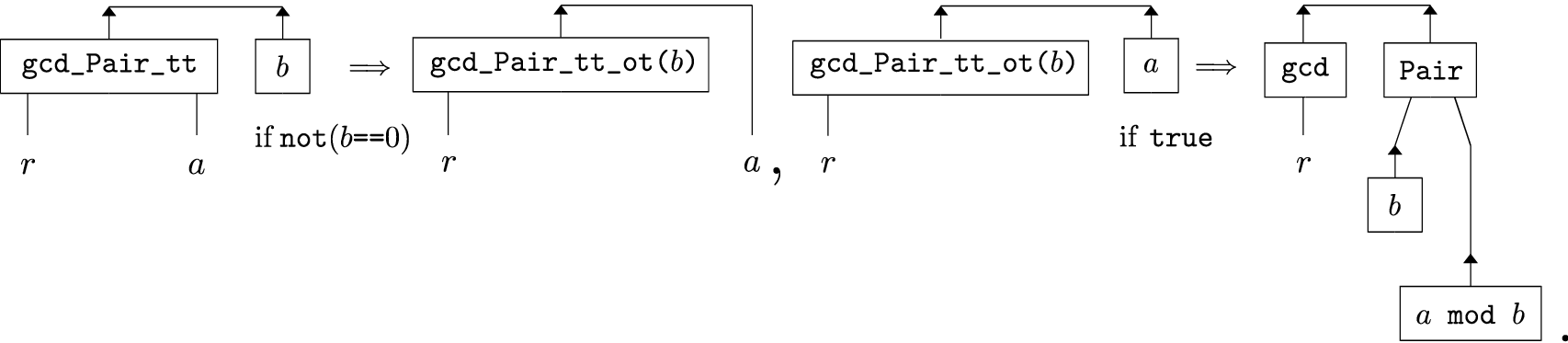}

\noindent
By the obtained conditional rules for $\mathtt{gcd\_Pair\_tt}\inter \mathtt{Int}(b)$,
we have two branches according to the value of $b$---return the result $a$ if $b$ is 0, otherwise proceed to the calculation.
We note the four rules obtained are the same as the non-nested rules in Example~\ref{example:gcd},
by changing $\mathtt{gcd\_Pair\_tt}$ and $\mathtt{gcd\_Pair\_tt\_ot}$
into $\mathtt{gcd1}$ and $\mathtt{gcd2}$, respectively.
\end{example}

For each connection of a conditional \textsc{nap}, 
the translation $\mathrm{T}$ introduces 
a rule in which one of the active pair agents
has a fresh symbol for $\Sigma$ and the previous occurring ones.
Therefore, the following holds:

\begin{lemma}\label{lemma:no-overlapped-in-a-rule}
  Let $R_\mathit{if}$ be a rule on a conditional \textsc{nap}.
  $\mathrm{T}[R_\mathit{if}]$ does not have
  two rules for the same active pair
 $P$ such that
  $(P \ \mathrm{if} \ c_1 \ito N_1)$ and 
  $(P \ \mathrm{if} \ c_2 \ito N_2)$.
  \qed
\end{lemma}

\begin{lemma}\label{lemma:no-overlapped-in-rules}
  Let $R_1, R_2$ be pairwise distinct rules on conditional \textsc{nap}s.
  When for the same active pair $P$,
  $(P \ \mathrm{if} \ c_1 \ito N_1)$ and
  $(P \ \mathrm{if} \ c_2 \ito N_2)$ are
  introduced by $\mathrm{T}[R_1]$ and $\mathrm{T}[R_2]$,
  then $c_1,c_2$ are disjoint or $N_1 = N_2$. \qed
\end{lemma}

We suppose an interaction net that has a symbol set $\Sigma$ and 
a set of pairwise distinct rules on conditional \textsc{nap}s $\mathcal{R}_\mathit{if}$.
By applying the translation $\mathrm{T}$ to each rule in 
$\mathcal{R}_\mathit{if}$, 
we obtain two sets $\mathcal{R}_\mathit{if}'$ and $\Sigma'$---a set of (non-nested) conditional rules and a set of symbol introduced during the translation,
respectively.
The system with a symbol set $\Sigma \cup \Sigma'$ and a rule set $\mathcal{R}_\mathit{if}'$ is also an interaction net
by Lemma~\ref{lemma:no-overlapped-in-a-rule} and \ref{lemma:no-overlapped-in-rules}.
In addition, the following proposition shows that
the conditional \textsc{nap} extension is conservative, 
\ie an interaction net with $\Sigma \cup \Sigma'$ and $\mathcal{R}_\mathit{if}$ is realised without the extension, thus as
an interaction net with $\Sigma \cup \Sigma'$ and $\mathcal{R}_\mathit{if}'$.

\begin{proposition}
  Let $(\nested{P_\mathit{if}} \ito N)$ be 
  a rule on conditional \textsc{nap}s and
  $\nested{P}$ be the condition dropped \textsc{nap} for 
  $\nested{P_\mathit{if}}$.
  If $\nested{P}$ is reduced to $N$ by using this rule,
  $\nested{P}$ is also reduced to $N$ by using only rules obtained from
  $\mathbf{T}[(\nested{P_\mathit{if}} \ito N)]$.
\end{proposition}

\begin{proof}
  By structural induction on conditional \textsc{nap}s.
  \begin{description}
  \item [Base case:]
    The translation $\mathrm{T}$ works as the identity function,
    so the $\nested{P}$ is reduced to $N$ by using
    rules obtained from the translation.

  \item[Step case:]
    We assume that $\nested{P_\mathit{if}}$ is
    $\nested{P_1 \ \mathrm{if} \ c_1,
      \ z-\gamma(\vec{w}) \ \mathrm{if} \ c, \ \vec{u}}$
      and the condition dropped \textsc{nap} $\nested{P}$
      is reduced to $N$ by the rule $(\nested{P_\mathit{if}} \ito N)$.
    $\mathrm{T}$ generates the following rules:
    \begin{description}
    \item[(1)] $(P_1 \ \mathrm{if} \ c_1 \ito M)$ where $M$ is a net
      obtained by connecting the $z$ to
    the principal port of an agent $\kappa(\vec{z})$,
    where the symbol $\kappa$ is a fresh and uniquely depending on the agents symbols in $P_1$ and the $c_1$.
    \item[(2)] $\mathrm{T}
      [\nested{\kappa(\vec{z})\inter\gamma(\vec{w}) \ \mathrm{if} \ c, \ \vec{u}}
        \ito N]$.
    \end{description}
    By the rule of (1), $P_1$ is reduced to
    an active pair $(\kappa, \gamma)$,
    and it is iteratively reduced to $N$ by the induction hypothesis. 
    \qedhere
  \end{description}
\end{proof}

\section{A notation for rules on conditional \textsc{nap}s}\label{sec:notation}
In this section, we introduce a notation
to facilitate expressing sets of pairwise distinct rules.

It is practical to write rules for an active pair $\alpha\inter\beta$
with conditional expressions $c_1,\ldots,c_n$ as \emph{guards}
like in Haskell~\cite{marlow2010haskell},
which are evaluated one at a time from $c_1$ to $c_n$, as follows:
\[
\begin{array}{l}
  \alpha(\vec{x}) \inter \beta(\vec{y}) \\
  \qquad \verb/| / c_1 \ito N_1\\
  \qquad \verb/| / c_2 \ito N_2\\
  \qquad\qquad \vdots\\
  \qquad \verb/| / c_n \ito N_n.\\
\end{array}
\]

\noindent
The last placed ``$\mathit{otherwise}$'' is evaluated as true. We can suppose that this is translated into the following rules
where the conditional expressions do not overlap:
\[\begin{array}{l}
((\alpha, \beta) \ \mathrm{if}\ c_1 \ {\ito} \ N_1), \\
((\alpha, \beta) \ \mathrm{if}\ \sym{not}(c_1) \ \sym{and} \ c_2 \ {\ito} \ N_2), \ \ldots,\\
  ((\alpha, \beta) \ \mathrm{if}\ \sym{not}(c_1 \ \sym{or} \ c_2 \ \sym{or} \ \cdots \ \sym{or} \ c_{n-1}) \ \sym{and} \ c_n \ {\ito} \ N_n).
\end{array}\]

\noindent
As another example, we take the following two rules
whose bases of conditional \textsc{nap}s are 
the same ones which are followed by connections to the same port:
\[
\begin{array}{l}
(\nested{\alpha(\vec{x})\inter\beta(\vec{y}) \ \mathrm{if}\ c,
  \ z-\gamma(\vec{z_1}) \ \mathrm{if}\ c_1} \ito N_1),\\
(\nested{\alpha(\vec{x})\inter\beta(\vec{y}) \ \mathrm{if}\ c,
  \ z-\kappa(\vec{z_2}) \ \mathrm{if}\ c_2} \ito N_2)
\end{array}
\]

\noindent
It is also convenient to group these connections by the same port $z$
with guards and a \texttt{case}-expression as follows:
\[
\begin{array}{l}
  \alpha(\vec{x}) \inter \beta(\vec{y}) \\
  \quad \verb/| / c \verb| -> case of | z\\
  \qquad\qquad\qquad
  \begin{array}{lcl}
    \gamma(\vec{z_1}) & \verb/|/ & c_1  \ \ito \ N_1\\
    \kappa(\vec{z_2}) & \verb/|/ & c_2 \ \ito \ N_2.\\
  \end{array}
\end{array}
\]

\noindent
We may omit the guards if there is only one condition $\texttt{true}$.

\begin{example}\label{example:gcd-case}
By using this notation,
the rules for $\texttt{gcd}\inter\texttt{Pair}$
in Example~\ref{example:gcd} is written
in the following single sentence:
\[
\begin{array}{l}
  \texttt{gcd}(r) \inter \texttt{Pair}(p_1,p_2) \\
  \quad \sym{-> case of } p_2\\
  \qquad\quad
  \begin{array}{lcl}
    \mathrm{Int}(b) & \sym{|} & b\sym{==}0 \ \ito \ N_1\\
                    & \sym{|} & \mathit{otherwise} \ \sym{-> case of } p_1\\
  \end{array}\\
  \qquad\qquad\qquad\qquad\qquad\qquad\qquad
  \begin{array}{l}
    \mathrm{Int}(a) \ \ito \ N_2.
  \end{array}
\end{array}
\]
\end{example}

\subsection{Definition and translations of the notation}
First, we define the rule notation by the following $\mathit{Rule}$
with $\mathit{Spray}$s, which are branches connected to nets:

\smallskip

\noindent{}
\begin{minipage}[t]{0.45\textwidth}
$
\setlength\arraycolsep{3pt}
\begin{array}{llllll}
  \mathit{Rule} &::=& \alpha(\vec{x}) \inter \beta(\vec{y}) & \sym{|} & c_1 & \mathit{Spray}_1\\
  &&&\sym{|} & c_2 & \mathit{Spray}_2\\
  &&& & \vdots \\
  &&&\sym{|} & c_m & \mathit{Spray}_m\\
\vspace{6em}
\end{array}
$
\end{minipage}
\hspace{0.04\columnwidth}
\begin{minipage}[t]{0.45\textwidth}
$
\setlength\arraycolsep{3pt}
\begin{array}{llllll}
  \mathit{Spray} &::=& \multicolumn{4}{l}{\ito N}\\
      &\mid & \multicolumn{4}{l}{\sym{->} \ \sym{case} \ z \ \sym{of} }\\
      &&\qquad \gamma_1(\vec{w_1}) & \sym{|} & c_{11} &  \mathit{Spray}_{11}\\
      && && \vdots & \\
      &&                          & \sym{|} & c_{1m_1} & \mathit{Spray}_{1m_1}\\
      &&\qquad\quad \vdots && \\
      &&\qquad \gamma_n(\vec{w_n}) & \sym{|} & c_{n1} & \mathit{Spray}_{n1}\\
      && && \vdots & \\
      &&                          & \sym{|} & c_{nm_n} & \mathit{Spray}_{nm_n}\\
\end{array}
$
\end{minipage}

\smallskip

\noindent
where $\alpha(\vec{x})$, $\beta(\vec{y})$ are agents,
$c_i$ is a conditional expression, $N$ is a net, $z$ is a free port,
$\gamma_1(\vec{w_1}), \ldots, \gamma_n(\vec{w_n})$ are
distinct agents,
and the conditional expression ``$\mathit{otherwise}$'' can appear last.
We may write $\textit{Spray}$ as $S$ when there is no confusion.

\paragraph{Translation of the rule notation to rules on conditional \textsc{nap}s:}
We define the following translations of the rule notation
to obtain rules on conditional \textsc{nap}s:

\[\mathrm{T_R}: \mathit{Rule} \mapsto \overrightarrow{\mathrm{T}_S[\nested{P_\mathit{if}}, \ \mathit{Spray}]}, \quad
\mathrm{T_S}: (\nested{P_\mathit{if}}, \ \mathit{Spray}) \mapsto \vec{r}
\]

where $\nested{P_\mathit{if}}$ is a conditional \textsc{nap},
$\vec{r}$ is a rule sequence.
If the $\nested{P_\mathit{if}}$ obtained during translation do not satisfy the conditions of given in Definition~\ref{def:conditional-naps}, then the translation will fail.

We may write 
``$\sym{not}(c_1 \ \sym{or} \ c_2 \ \sym{or} \ \cdots \ \sym{or} \ c_{n-1}) \ \sym{and} \ \ c_n$''
as 
``$\sym{not}(c_1, \ c_2, \ldots, \  c_{n-1}) \ \sym{\&} \ c_n$''
to save space.
In addition, we may write ``$\sym{not}(c_1, \ c_2, \ldots, \  c_{n-1}) \ \sym{\&} \ \mathit{otherwise}$'' just as ``$\mathit{otherwise}$''.

\begin{flalign*}
  &
  \begin{array}{lcl}
  \mathrm{T_R}
  \left[
  \setlength\arraycolsep{3pt}
  \begin{array}{llll}
  \alpha(\vec{x}) \inter \beta(\vec{y}) & \sym{|} & c_1 & \mathit{S}_1\\
  &\sym{|} & c_2 & \mathit{S}_2\\
  & & \vdots \\
  &\sym{|} & c_n & \mathit{S}_n\\
\end{array}
\right]
  &\stackrel{\mathrm{def}}{=}&  
  \begin{array}{l}
    \mathrm{T_S}[\nested{\alpha(\vec{x}) \inter \beta(\vec{y}) \ \mathrm{if}\ c_1},
      \ \mathit{S}_1],\\
    \mathrm{T_S}[\nested{\alpha(\vec{x}) \inter \beta(\vec{y}) \ \mathrm{if} \ \sym{not}(c_1) \ \sym{\&} \ c_2}, 
      \ \mathit{S}_2]\\
    \multicolumn{1}{c}{\vdots}\\
    \mathrm{T_S}[\nested{\alpha(\vec{x}) \inter \beta(\vec{y}) \ \mathrm{if} \ \sym{not}(c_1, \ c_2, \ldots, \ c_{n-1}) \ \sym{\&} \ c_n }, 
      \ \mathit{S}_n].\\
  \end{array}
  \end{array}&
\end{flalign*}


\begin{flalign*}
  &
  \setlength\arraycolsep{2pt}
  \begin{array}{lcl}
    \mathrm{T_S} \left[ \nested{P_\mathit{if}},\  {\ito N} \right]
    &
    \stackrel{\mathrm{def}}{=}
    &
    (\nested{P_\mathit{if}} \ito N),  \\\\
    %
    %
    \mathrm{T_S}
    \left[
      \setlength\arraycolsep{2pt}
      \nested{P_\mathit{if}},
      \begin{array}{llll}
        \multicolumn{4}{l}{\sym{->} \ \sym{case} \ z \ \sym{of} }\\
        \qquad \gamma_1(\vec{w_1}) & \sym{|} & c_{11} &  \mathit{S}_{11}\\
        && \vdots & \\
        & \sym{|} & c_{1m_1} & \mathit{S}_{1m_1}\\
        \qquad\quad \vdots && \\
        \qquad \gamma_n(\vec{w_n}) & \sym{|} & c_{n1} & \mathit{S}_{n1}\\
        && \vdots & \\
        & \sym{|} & c_{nm_n} & \mathit{S}_{nm_n}\\
      \end{array}
      \right]
    &\stackrel{\mathrm{def}}{=}&
    \begin{array}{l}
      \mathrm{T_S}[\nested{P_\mathit{if},\ z-\gamma_1(\vec{w_1}) \ \mathrm{if} \ c_{11}}, 
        \ \mathit{S}_{11}],\\
      \multicolumn{1}{c}{\vdots}\\
      \mathrm{T_S}[\nested{P_\mathit{if},\ z-\gamma_1(\vec{w_1}) \ \mathrm{if} \ 
          \sym{not}(c_{11}, \ldots, \ c_{n{m_1-1}}) \ \sym{\&} \ c_{1m_1}}, 
        \ \mathit{S}_{1m_1}],\\
      \multicolumn{1}{c}{\vdots}\\
      \mathrm{T_S}[\nested{P_\mathit{if},\ z-\gamma_n(\vec{w_n}) \ \mathrm{if} \ c_{n1}}, 
        \ \mathit{S}_{n1}],\\
      \multicolumn{1}{c}{\vdots}\\
      \mathrm{T_S}[\nested{P_\mathit{if},\ z-\gamma_1(\vec{w_n}) \ \mathrm{if} \ 
          \sym{not}(c_{n1}, \ldots, \ c_{n{m_n-1}}) \ \sym{\&} \ c_{nm_n}}, 
        \ \mathit{S}_{nm_n}].\\
    \end{array}
  \end{array}
  &
\end{flalign*}

\begin{example}\label{example:trans-gcd-case}
  The following is a translation
  of the rule notation in Example~\ref{example:gcd-case}:
\begin{flalign*}
&\mathrm{T_R}
\left[
  \setlength\arraycolsep{2pt}
\begin{array}{l}
  \texttt{gcd}(r) \inter \texttt{Pair}(p_1,p_2) \\
  \quad \sym{| true -> case of } p_2\\
  \qquad\qquad\qquad\quad
  \begin{array}{lcl}
    \mathrm{Int}(b) & \sym{|} & b\sym{==}0 \ \ito \ N_1\\
                    & \sym{|} & \mathit{otherwise} \ \sym{-> case of } p_1\\
  \end{array}\\
  \qquad\qquad\qquad\qquad\qquad\qquad\qquad\qquad\quad
  \begin{array}{l}
    \mathrm{Int}(a) \ \sym{| true} \ \ito \ N_2
  \end{array}
\end{array}  
\right]&
\end{flalign*}
%
%
\begin{flalign*}
&=\mathrm{T_S}
  \left[
    \nested{\texttt{gcd}(r) \inter \texttt{Pair}(p_1,p_2) \ 
      \mathrm{if} \ \sym{true}}, \
    \setlength\arraycolsep{2pt}
    \begin{array}{l}
      \sym{-> case of } p_2\\
      \qquad
      \begin{array}{lcl}
        \mathrm{Int}(b) & \sym{|} & b\sym{==}0 \ \ito \ N_1\\
        & \sym{|} & \mathit{otherwise} \ \sym{-> case of } p_1\\
      \end{array}\\
      \qquad\qquad\qquad\qquad\qquad\qquad
      \begin{array}{l}
        \mathrm{Int}(a) \ \sym{| true} \ \ito \ N_2
      \end{array}
    \end{array}
\right]&
\end{flalign*}
%
%
\begin{flalign*}
  &=\mathrm{T_S}
  \left[
    \nested{\texttt{gcd}(r) \inter \texttt{Pair}(p_1,p_2) \ 
      \mathrm{if} \ \sym{true}, \ p_2 - \mathrm{Int}(b) \
      \mathrm{if} \ b\sym{==}0}, \ \ito \ N_1
    \right],&\\
  &\quad \ 
  \mathrm{T_S}
  \left[
    \nested{\texttt{gcd}(r) \inter \texttt{Pair}(p_1,p_2) \ 
      \mathrm{if} \ \sym{true}, \ p_2 - \mathrm{Int}(b) \
      \mathrm{if} \ \mathit{otherwise}}, \ 
    \setlength\arraycolsep{3pt}
    \begin{array}{l}
      \sym{-> case of } p_1 \\
      \qquad
      \begin{array}{lcl}
        \mathrm{Int}(a) & \sym{|} & \sym{true} \ \ito \ N_2
      \end{array}
    \end{array}
    \right]& 
\end{flalign*}
%
%
\begin{flalign*}
  &=(
    \nested{\texttt{gcd}(r) \inter \texttt{Pair}(p_1,p_2) \ 
      \mathrm{if} \ \sym{true}, \ p_2 - \mathrm{Int}(b) \
      \mathrm{if} \ b\sym{==}0} \ \ito \ N_1
    ),&\\
  &\quad \ 
  \mathrm{T_S}
  \left[
    \nested{\texttt{gcd}(r) \inter \texttt{Pair}(p_1,p_2) \ 
      \mathrm{if} \ \sym{true}, \ p_2 - \mathrm{Int}(b) \
      \mathrm{if} \ \mathit{otherwise},
      \ p_1 - \mathrm{Int}(a) \ \mathrm{if} \ \sym{true}}, \ \ito \ N_2
    \right]&
\end{flalign*}
%
%
\begin{flalign*}
  &=(
    \nested{\texttt{gcd}(r) \inter \texttt{Pair}(p_1,p_2) \ 
      \mathrm{if} \ \sym{true}, \ p_2 - \mathrm{Int}(b) \
      \mathrm{if} \ b\sym{==}0} \ \ito \ N_1
    ),&\\
  &\quad \ 
  (
    \nested{\texttt{gcd}(r) \inter \texttt{Pair}(p_1,p_2) \ 
      \mathrm{if} \ \sym{true}, \ p_2 - \mathrm{Int}(b) \
      \mathrm{if} \ \mathit{otherwise},
      \ p_1 - \mathrm{Int}(a) \ \mathrm{if} \ \sym{true}} \ \ito \ N_2
    ).
    & 
\end{flalign*}
Thus, we derive the same rules as in Example~\ref{example:gcd}.
\end{example}

\subsection{Properties of the rule notation}
In this section we show properties of the notation. 
First, we show that the translation $\mathrm{T_R}$ is a function
if $\mathrm{T_R}$ does not fail,
thus we have the same result for the same rule notation.

\begin{lemma}\label{lemma:Uniqueness-of-S-for-the-same-Pif-in-Ts}
  Let $\mathrm{Rule}$ be a string accepted by the rule notation,
  and let $\mathrm{T_R}[\mathrm{Rule}]$ not fail.
  If during expansions of $\mathrm{T_R}[\mathrm{Rule}]$ to rules, 
  the followings occur for the same conditional \textsc{nap} $\nested{P_\mathit{if}}$,
  then $S_1$ and $S_2$ are the same spray:
  \begin{itemize}
    \item $\mathrm{T_S}[\nested{P_\mathit{if}}, \ S_1]$,
    \item $\mathrm{T_S}[\nested{P_\mathit{if}}, \ S_2]$.
  \end{itemize}
\end{lemma}

\begin{proof}
  By structural induction on the conditional \textsc{nap}s $\nested{P_\mathit{if}}$.
%
\end{proof}

\noindent
From now on, we will only consider the case where the translation $\mathrm{T_R}$ does not fail.
By Lemma~\ref{lemma:Uniqueness-of-S-for-the-same-Pif-in-Ts},
we can show the property more precisely:

\begin{lemma}\label{lemma:Sequentiality-of-Ts}
  Let $\mathrm{Rule}$ be a string accepted by the rule notation.
  If during expansions of $\mathrm{T_R}[\mathrm{Rule}]$ to rules,
  the followings occur for the same conditional \textsc{nap} $\nested{P_\mathit{if}}$,
  then $y_1$ and $y_2$ are the same:
  \begin{itemize}
  \item $\mathrm{T_S}[\nested{P_\mathit{if}, \ y_1-\gamma_1(\vec{w}_1) \ \mathrm{if} \ c_1}, \ S_1]$,
  \item $\mathrm{T_S}[\nested{P_\mathit{if}, \ y_2-\gamma_2(\vec{w}_2) \ \mathrm{if} \ c_2}, \ S_2]$.
  \end{itemize}
  In addition,
  when $\gamma_1(\vec{w}_1)$ and $\gamma_2(\vec{w}_2)$ are the same,
  $c_1$ and $c_2$ are not overlapped.
\end{lemma}

\begin{proof}
By the definition of $\mathrm{T_S}$,
$\mathrm{T_S}[\nested{P_\mathit{if}, \ y_1-\gamma_1(\vec{w}_1) \ \mathrm{if} \ c_1}, \ S_1]$ and
$\mathrm{T_S}[\nested{P_\mathit{if}, \ y_2-\gamma_2(\vec{w}_2) \ \mathrm{if} \ c_2}, \ S_2]$
are directly derived from
$\mathrm{T_S}[\nested{P_\mathit{if}}, \ S'_1]$ and
$\mathrm{T_S}[\nested{P_\mathit{if}}, \ S'_2]$ for some $S'_1$ and $S'_2$, respectively.
By Lemma~\ref{lemma:Uniqueness-of-S-for-the-same-Pif-in-Ts}, 
$S'_1 = S'_2$, and therefore this lemma holds.
\end{proof}

Here we form a set whose elements are conditional \textsc{nap}s
that appear during expansions of $\mathrm{T_R}$ to rules,
and show that the set is sequential.

\begin{definition}
  For conditional \textsc{nap}s $\nested{P_\mathit{if}}$ and
  $\nested{P_\mathit{if}, \ \vec{u}}$ for any sequence of connections $\vec{u}$
  (which may be empty),
  we write $\nested{P_\mathit{if}} \subseteq \nested{P_\mathit{if}, \ \vec{u}}$.
  We use $\mathrm{AllSub}(\nested{P_\mathit{if}})$ as a set
  $\{P'_\mathit{if} \ \mid \  P'_\mathit{if} \subseteq P_\mathit{if}  \}$.
\end{definition}

\begin{example}
  The followings are the results of applying 
  $\mathrm{AllSub}$ to rules obtained in Example~\ref{example:trans-gcd-case}:
  \begin{itemize}
  \item Elements of
    $\mathrm{AllSub}(\nested{\texttt{gcd}(r) \inter \texttt{Pair}(p_1,p_2) \ 
      \mathrm{if} \ \sym{true}, \ p_2 - \mathrm{Int}(b) \
      \mathrm{if} \ b\sym{==}0})$:
    \begin{itemize}
      \item $\nested{\texttt{gcd}(r) \inter \texttt{Pair}(p_1,p_2) \ 
        \mathrm{if} \ \sym{true}}$,
        \item $\nested{\texttt{gcd}(r) \inter \texttt{Pair}(p_1,p_2) \ 
      \mathrm{if} \ \sym{true}, \ p_2 - \mathrm{Int}(b) \
      \mathrm{if} \ b\sym{==}0}$.
    \end{itemize}
  \item Elements of
    $\mathrm{AllSub}(\nested{\texttt{gcd}(r) \inter \texttt{Pair}(p_1,p_2) \ 
      \mathrm{if} \ \sym{true}, \ p_2 - \mathrm{Int}(b) \
      \mathrm{if} \ \mathit{otherwise},
      \ p_1 - \mathrm{Int}(a) \ \mathrm{if} \ \sym{true}})$: 
    \begin{itemize}
    \item $\nested{\texttt{gcd}(r) \inter \texttt{Pair}(p_1,p_2) \ 
      \mathrm{if} \ \sym{true}}$,
    \item $\nested{\texttt{gcd}(r) \inter \texttt{Pair}(p_1,p_2) \ 
      \mathrm{if} \ \sym{true}, \ p_2 - \mathrm{Int}(b) \
      \mathrm{if} \ \mathit{otherwise}}$,
    \item $\nested{\texttt{gcd}(r) \inter \texttt{Pair}(p_1,p_2) \ 
      \mathrm{if} \ \sym{true}, \ p_2 - \mathrm{Int}(b) \
      \mathrm{if} \ \mathit{otherwise},
      \ p_1 - \mathrm{Int}(a) \ \mathrm{if} \ \sym{true}}$.
    \end{itemize}
      
  \end{itemize}
  
\end{example}

\begin{proposition}\label{proposition:sequentiality-of-Tr}
  Let $\mathrm{Rule}$ be a string accepted by the rule notation.
  When we obtain rules $(\nested{P_{{\mathit{if}_1}}} \ito N_1), \ldots, \
  (\nested{P_{{\mathit{if}_m}}} \ito N_m)$ from $\mathrm{T_R}[\mathrm{Rule}]$,
  then
  $\mathrm{AllSub}(\nested{P_{{\mathit{if}_1}}}) \ \cup \ \cdots \ \cup \ 
  \mathrm{AllSub}(\nested{P_{{\mathit{if}_m}}})$ is sequential.
\end{proposition}

\begin{proof}
We examine whether the set
$\mathcal{A}=\mathrm{AllSub}(\nested{P_{{\mathit{if}_1}}}) \ \cup \ \cdots \ \cup \ 
  \mathrm{AllSub}(\nested{P_{{\mathit{if}_m}}})$
satisfies sequentiality conditions in
Definition~\ref{definition:sequentiality-of-cnap}.

First,
we check the condition (1a).
Each $\nested{P_{{\mathit{if}_i}}}$ is obtained 
from the same rule notation $\mathrm{Rule}$ by using the translation of $\mathrm{T_R}$.
Therefore, when 
$\nested{\alpha(\vec{x}) \inter \beta(\vec{y}) \ \mathrm{if}\ \mathit{c}_1},
\nested{\alpha(\vec{x}) \inter \beta(\vec{y}) \ \mathrm{if}\ \mathit{c}_2} \in \mathcal{A}$,
$c_1$ and $c_2$ are disjoint, by the definition of $\mathrm{T_R}$.

Next, we verify conditions in (2).
We suppose that 
$(\nested{\alpha(\vec{x}) \inter \beta(\vec{y}) \ \mathrm{if}\ \mathit{c}, 
\ u_1, \ u_2, \ldots, \ u_n} \ito N)$ is derived from 
$\mathrm{T_R}[\mathrm{Rule}]$.
All elements of $\mathrm{AllSub}(\nested{\alpha(\vec{x}) \inter \beta(\vec{y}) \ \mathrm{if}\ \mathit{c}, 
\ u_1, \ u_2, \ldots, \ u_n})$ 
are related by $\subseteq$ as follows:
$$\nested{\alpha(\vec{x}) \inter \beta(\vec{y}) \ \mathrm{if}\ \mathit{c}}
\subseteq 
\nested{\alpha(\vec{x}) \inter \beta(\vec{y}) \ \mathrm{if}\ \mathit{c}, \ u_1}
\subseteq \cdots \subseteq 
\nested{\alpha(\vec{x}) \inter \beta(\vec{y}) \ \mathrm{if}\ \mathit{c}, 
\ u_1, \ u_2, \ldots, \ u_n}.$$
By the definition of $\mathrm{T_R}$,
for each two elements such that $\nested{P_\mathit{if}} \subseteq \nested{P_\mathit{if}, \ u}$
in the $\mathrm{AllSub}$,
there is the following expansion:
$$\mathrm{T_S}[\nested{P_\mathit{if}}, \ S] =
      \ldots, \ \mathrm{T_S}[\nested{P_\mathit{if}, \ u}, \ S'], \ \ldots.$$
Thus every element
must occur as the first argument of $\mathrm{T_S}$ during the applications of $\mathrm{T_R}$. 
Therefore, by Lemma~\ref{lemma:Sequentiality-of-Ts},
each $\mathrm{AllSub}(\nested{P_{{\mathit{if}_i}}})$ satisfies (2a) and (2c).
(2b) also is satisfied by the definition of $\mathrm{AllSub}$.
$\mathrm{AllSub}(\nested{P_{{\mathit{if}_i}}})$
and $\mathrm{AllSub}(\nested{P_{{\mathit{if}_j}}})$ are disjoint
because these are derived for the same $\mathrm{Rule}$ by $\mathrm{T_R}[\mathrm{Rule}]$.
Therefore,
$\mathcal{A}$ also satisfies all these conditions.
\end{proof}

We suppose that
$\mathcal{R_\mathit{if}}$ is a rule set whose elements are 
obtained from the same rule notation $\mathrm{Rule}$
by using $\mathrm{T_R}[\mathrm{Rule}]$.
Rules in $\mathcal{R_\mathit{if}}$ do not have any overlapped conditions,
and thus there are not any rules $(\nested{P_\mathit{if}} \ito N)$
such that the $\nested{P_\mathit{if}}$ is a subnet of
$\nested{P'_\mathit{if}}$ for other rules $(\nested{P'_\mathit{if}} \ito N')$
in $\mathcal{R_\mathit{if}}$.
Therefore, as long as we use the rule notation,
the rule set derived from $\mathcal{R_\mathit{if}}$ will be 
pairwise distinct.

\begin{proposition}\label{proposition:well-formedness-of-Tr}
  Let $\mathrm{Rule}$ be a string accepted by the rule notation,
  and $\mathcal{R_\mathit{if}}$ a set
  whose elements are obtained from $\mathrm{T_R}[\mathrm{Rule}]$.
  Then, $\mathcal{R_\mathit{if}}$ is pairwise distinct. \qed
\end{proposition}

\section{Discussion}\label{sec:discussion}

\paragraph{Implementation:}
Using rules on conditional \textsc{nap}s allows us to write algorithms naturally and reduces the number of reduction steps in comparison with using non-nested ones.
However, when it comes to the implementation of these rules, 
we need an efficient matching mechanism so that the rule application
waits until each nested agent is connected to all its ports and avoids repeatedly checking that these connections have been made.
This method can be realised by the $\mathrm{T}$ translation, 
since it ensures that the nested agents are connected.
Given this, 
the matching can then be performed by the net evaluator.
Another advantage of using the $\mathrm{T}$ translation
is that it is then possible to perform some rule optimisation, by combining rules.
For instance, we take the translated rules in Example~\ref{example:trans-gcd}.
The RHS of the rule $\mathtt{gcd\_Pair\_tt\_ot}(b) \inter \mathtt{Int}(a)$
has an agent pair $\mathtt{gcd} \inter \mathtt{Pair}$.
Thus, by applying a rule for $\mathtt{gcd} \inter \mathtt{Pair}$ to the RHS,
we can have a one-step reduced optimisation rule as follows:
\begin{center}
\includegraphics[scale=\smallscale,keepaspectratio,clip]{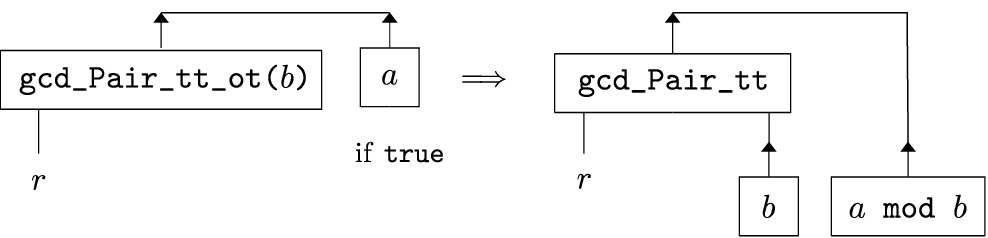}
\end{center}

\noindent{}Therefore we can say that by using the translation T, conditional \textsc{nap}s can be realised, taking advantage of the expressiveness of 
the case notation and the already implemented evaluator
such as Inpla~\cite{DBLP:phd/ethos/Sato15}
\footnote{An up-to-date implementation is available from \url{https://github.com/inpla/inpla}}.

\paragraph{Related work:} 
It is also possible to realise nested pattern matching for pure interaction nets 
by using other approaches~\cite{DBLP:conf/sas/Bechet92,DBLP:journals/entcs/SinotM05} where agents are allowed to have more than one principal port.
However, there are some limitations, as discussed in \cite{DBLP:journals/entcs/HassanS08}.
For example, there is no natural way to express the function $\texttt{lastElt}$, which returns the last element of a list. 
The last element of a list $\texttt{[} x \texttt{]}$ is 
matched by using a pattern,
which is the LHS of a rule,
with $\mathtt{Cons}$ agent having two principal ports 
and $\mathtt{Nil}$ agent. However, 
to keep the one-step confluence, 
the patterns must be such that 
every principal port is connected to the principal port of another agent.
Thus, because two principal ports of $\mathtt{Cons}$ cannot be free,
we need to have every pattern for lists 
such as $\texttt{[} y_1,x \texttt{]}, \texttt{[}y_1,y_2,x \texttt{]}, \ldots$.
\begin{center}
\includegraphics[scale=\tinyscale,keepaspectratio,clip]{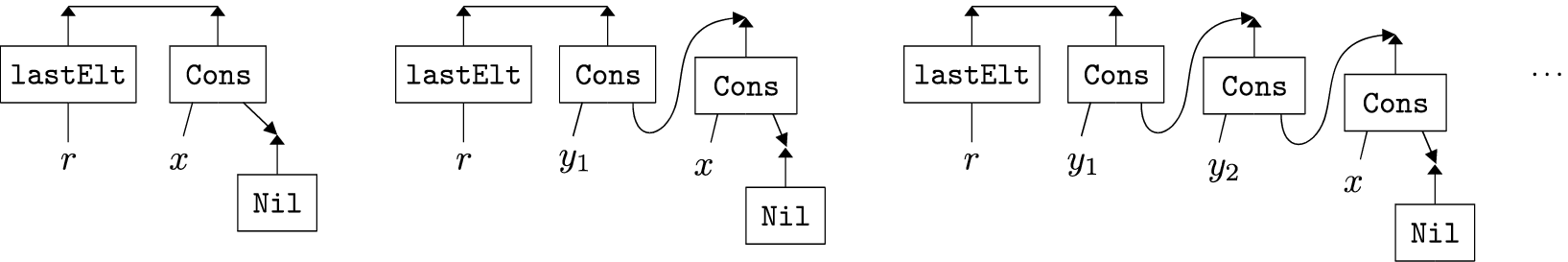}
\end{center}

\section{Conclusion}\label{sec:conclusion}

In this paper we introduced conditional nested pattern matching
as an extension of \cite{DBLP:journals/entcs/HassanS08},
and showed that as long as a rule set is pairwise distinct,
the rewritings can continue to be one-step confluent.
This not only allows programs to be more easily written in interaction net form, but also allows the execution to be performed by existing evaluators. Most well-known algorithms contain computations of conditional nested pattern matching, and these can now be realised in interaction nets. We expect this will contribute to finding new ways of implementing such algorithms, taking advantage of the benefits of interaction nets, such as built-in parallelism
and internal garbage collection.

\bibliography{bibfile}

\end{document}